\documentclass[a4paper]{article}

\usepackage{tcolorbox}

\usepackage{url}
\usepackage{amsmath}
\usepackage{amsthm}

\newtheorem{theorem}{\bf Theorem}[section]
\newtheorem{condition}{\bf Condition}[section]

\begin{document}

\title{Observement as Universal Measurement}

\author{
David G. Green\footnote{Faculty of Information Technology, Monash University, Clayton, Victoria, Australia.} \footnote{Corresponding Author. Email: david.green@monash.edu},
Kerri Morgan\footnote{School of Information Technology, Faculty of Science Engineering \& Built Environment, Deakin University, Geelong, Australia.},
and Marc Cheong\footnote{Centre for AI and Digital Ethics (CAIDE), School of Computing and Information Systems, University of Melbourne, Parkville, Australia.}
}

\maketitle

\begin{abstract}
Measurement theory is the cornerstone of science, but no equivalent theory underpins the huge volumes of non-numerical data now being generated. In this study, we show that replacing numbers with alternative mathematical models, such as strings and graphs, generalises traditional measurement to provide rigorous, formal systems (`observement') for recording and interpreting non-numerical data. Moreover, we show that these representations are already widely used and identify general classes of interpretive methodologies implicit in representations based on character strings and graphs (networks). This implies that a generalised concept of measurement has the potential to reveal new insights as well as deep connections between different fields of research.
\end{abstract}

{\bf Subjects:} complexity, mathematical modelling

{\bf Keywords:} formal languages, graph theory, measurement, non-numeric data, observement, complexity

\section{Introduction}\label{secIntro}

It is impossible to overstate the influence that measurement has exerted on scientific thinking. 
In physics, for instance, both theory and experiment are dominated by concepts that are expressed as numerical values: time, distance, mass, charge, etc. 
Such measures are so familiar we hardly 
think about their implications; and yet they define fundamental concepts. 
In the 20th Century, quantitative thinking came to play a role in almost every field of science.  

The association of the term `measurement' with numbers is so deeply ingrained it 
blinds us to fundamental benefits that the \textit{process} of measurement itself conveys. 
A formal system of measurement performs several crucial roles in research \cite{tal2017sep}.

\maketitle

\begin{enumerate}
\item \emph{It ensures that data are gathered in a standard way.} The application of standards ensure that measurements are taken in a consistent and comparable form. 
This makes it possible to compare and combine data from  different sources. 
\item \emph{It produces data with well-known properties.} 
Representing attributes as numbers means that we can express  relationships between values using equations and other well-known tools. 
\item \emph{It has the power of mathematical abstraction. } 
Representing concepts as numbers makes it possible to develop mathematical models and techniques that apply to a wide range of phenomena. 
\item \emph{It shapes the development of theory and methods.} 
Combined with the power of quantitative mathematics, measurable variables (such as mass and length) shape the way we think about the world.
\end{enumerate}

Measurement is an indispensable cornerstone of science, but its success may limit 
our thinking. It has led to a culture where 
attempts are made to reduce problems and phenomena to numbers.
This approach may oversimplify or bias thinking: the quality of human performance is more than the value of (say) dollar profit per quarter; avian behavior is more complex than the number of times birds visit their nest. 
Needless to say, socio-cultural and historical contexts, and nuances of human behaviour, are disregarded within this approach.

A consequence of the information revolution is that organisations collect
enormous volumes of data, and much of it is non-numerical in nature. 
However, theoretical foundations to underlie current practices in data collection and analysis have lagged behind. 
The many advantages of formal measurement suggest that there are 
benefits to be gained by extending its underlying principles 
to non-numerical data. 
In this study, we show that replacing arithmetic with other mathematical models 
in the definition of measurement provides formal systems for representing non-numerical data. 
 
Because the term `\textit{measurement}' is so closely associated with the obtaining of, and assignment to numbers
\cite{vonHelmholtz1895zahlen,campbell1920,tal2017sep}, 
we make the distinction clear by introducing the term 
\emph{observement} for formal systems of observation.
In our definition (Section \ref{formal definition}), observement is a formal system that maps real-world phenomena to well-defined representations. 
In this sense, measurement is a special case, or subset of observement, in which the representation is a number. 

We argue that many kinds of observations could, or already can be, considered as observement. To illustrate how it works in practice, we present two examples as case studies. These are based on abstract models 
that are widely-used for representing non-numeric data: 
strings of symbols and graphs (networks). 
A variety of analytic techniques have supported these representations --  many of them common to both -- as well as rapidly-developing bodies of supporting theory. 


\section{Measurement as Formal Observation}
\label{secMeasurementFormalObservation}
The theory of measurement arose from the need to codify and standardise 
procedures for representing properties by numerical values. 
Its development as a formal approach to science dates back at least to the 
late 17th century and John Locke's interest in metrology \cite{anstey2016locke}. 

Perhaps the most general approach to measurement is the \emph{Representation Theory of Measurement} (RTM). 
It assumes that a system assigns \textit{numbers} to objects in such a way that they \emph{represent} a particular property. 
Formally, there is a \textbf{\textit{measurement system}} $(\mathcal{S}, \mathcal{N}, M)$   where the system 
$\mathcal{S}=\langle S,R\rangle$ is a set $S$ of objects and a finite set $R$ of relations on $S$,  
and 
$\mathcal{N} = \langle N,P\rangle$  
is the corresponding numerical system, consisting of numbers $N$ and 
a finite set of relations $P$ on $N$. $M$ is a mapping from $S$ to $N$. 
\begin{tcolorbox}
For mass, for instance, the system is comprised of the set of all physical objects $S$ and the set of relations between them, such as ``heavier''. 
The numerical system would comprise non-negative real numbers and relations such as $\geq$. 
\end{tcolorbox}

The RTM specifies the following three conditions that a measure must satisfy 
 \cite{finkelstein1984review}, \cite{michell2007representational}: 
 
\begin{condition}
     \textbf{(M1) Representation} -- there is an experimental process (mapping) that defines a homomorphism (property-preserving map) from objects to numbers.
\end{condition}

Formally, measurement is an experimental process that generates a mapping $M: S \rightarrow N$. 
Moreover $M$ must be a homomorphism, 
that is, for any relation $r \in R$, there is a corresponding relation 
$p \in P$, such that for any $x_{1}, x_{2}, \ldots x_{k} \in S$:
$r(x_{1}, x_{2}, \ldots x_{k}) \Leftrightarrow p(M(x_{1}), M(x_{2}), \ldots, M(x_{k})). $  The \textit{representation condition} requires that a set $H(S)$ of homomorphisms from $S$ to $N$ can be proved to exist.
\begin{tcolorbox}
For example, if object $a$ is \textit{heavier} than object $b$, 
then the measurements of their masses, $m(a)$ and $m(b)$, must satisfy $m(a) \geq m(b)$. 

\end{tcolorbox}

\begin{condition}
    \textbf{(M2) Existence} -- there must be at least one mapping.
\end{condition}

Clearly, there must be at least one mapping from $S$ to $N$, otherwise there exists no process to ``measure'' objects in $S$.  This leads to the existence condition: 
the set $H(S)$ of homomorphisms from $S$ to $N$ is non-empty. 

\begin{tcolorbox}
For example, procedures for measuring mass ensure that a measure $m$ exists. 
\end{tcolorbox}

\begin{condition}
    \textbf{(M3) Uniqueness} -- any two mappings that satisfy Condition (M1) are equivalent up to homomorphism. 
\end{condition}

This conditions means that measure must be unique in the sense that 
any two measures of the same object must be related, 
that is, if there exists another measure  
$M' : S \rightarrow  N$, 
then $M$ and $M'$ are related by an isomorphism $f : M \rightarrow M'$ such that for any relation $r\in R$ and corresponding relation $p \in P$, 
$ r(x_{1},x_{2},\dots, x_{k}) \Leftrightarrow  p(f(M(x_{1})), f(M(x_{2})),  \ldots, f(M(x_{k}))). $
\begin{tcolorbox}
For example, the mass of an item can be measured in kilograms, $M_{K}$, or pounds, $M_{P}$, but there is a simple transformation $M_{K} \rightarrow M_{P} \times 0.453592$ for converting one to the other. 
\end{tcolorbox}

Procedures for measuring mass ensure that a measure $m$ exists, and 
that the measurement is unique: that is, there is a simple conversion between two different measurement algorithms 
(say) pounds and kilograms. 

One advantage of measurement is abstraction of  entities and properties (and the relations between them) 
to a simple representation: a number (and the corresponding relations between numbers). 
This made it possible for humans to think and work with abstract ideas. 
However, mathematics provides many other models that allow us to represent 
other features of the real world in abstract terms. 
Graphs (networks), for instance, serve as abstract representations for many complex structures (Section \ref{secGraphs}).

Traditionally, measurement deals with numbers, but the essence of measurement is not quantities; it is the rigorous process by which observations are obtained 
\cite{weyl1949philosophy}, 
\cite{mari2017quantities}.
The RTM 
 specifies the features that every measurement system must have. 
However, RTM  
works just as well as a formal system if the underlying 
model is not arithmetic, but another representation model
instead.
This means that we can define formal systems of observation for 
any representation that is based on an underlying formal model.

In order to satisfy the RTM conditions there must be a procedure (an `\textit{algorithm}') that implements the mapping from  real world objects to the abstract model. 
Traditionally, this algorithm is built upon standards that define how to interpret the 
constants and relations in the model. 
For measurement, the standards would typically define what is meant by numbers, $0$ and $1$; 
as well as functions such as `$+$'; and relations, such as `$=$'. 
But such an algorithm need not necessarily output a number. For example, an algorithm might 
output a string of symbols. Our concept of observement builds on this idea of generalising the 
types of outputs from the process as part of the generalisation of measurement to observement.

The use of standards has never been confined to traditional measurement. 
Standards today are used for everything from food preparation to business management, 
from design of petrol pumps to building safety. 
The International Standards Organization \cite{iso2018} maintains thousands of standards for a vast range of items and procedures. 
These standards often include rules for representing data \cite{ansi2018}. 
They are aided by increasing use of automation to acquire, manipulate and display data.

\section{A General Definition of Observement}\label{secAGeneralDefinitionofObservement}

\subsection{Observement as a Generalisation of Measurement}
\label{secGeneralisationMeasurement}
Observement partly subscribes to the realist philosophy of measurement, which ``distinguishes what is measured from how it is measured; [and] holds  that  what  is  measured  are  attributes  of things,  rather  than  things  themselves'' \cite{michell2004}.

As ``numbers are in no way the only usable symbols'' \cite{weyl1949philosophy},  we propose an extension of the representational view of measurement, to remove the dependence -- inherent in measurement -- ``upon an isomorphism between an empirical system and a numerical system'' \cite{michell2004}.  Observement still supports the idea of the ``mapping, of relations between objects and ... entities'' \cite{tal2017sep}, and the traditional notion of numerical measurement can be seen as a subset of observement which designates `numbers' as such entities.

Our definition of observement ties in better with existing concepts of `measuring' systems such as social networks \cite{marsden1990} which need not be reduced to numbers, as they would be (intuitively) better represented as a network graph (Section \ref{secGraphs}).

By reducing measurement to numbers, we abstract away details ranging from temporal detail \cite{Gysi2020-RoyalSociety}  to historical and cultural context per Jos\'{e} Ortega y Gasset's conception of \textit{perspectivism} \cite{holmes2017}. The subjectivity of the observer's ``needs that [actually] ... interpret the world''  \cite{nietzsche1968}; an example is how a social network captures social and cultural context more intuitively than a mere reduction-to-numbers. The \textit{Graph Commons} project ``makes the computation and visualisation of network mappings \textit{accessible in a way that does not rely on mathematical ability} ... as a way to help reshape \textit{shared understandings} in the context of an \textit{active social struggle}'' (\textit{emphases ours})  \cite{McQuillan2018,arikan2016}.


\subsection{Formal Definition of Observement}
\label{formal definition}
We introduce the following definition of observement. In this definition, we generalise Conditions \textbf{M1}--\textbf{M3} in 
RTM (Section \ref{secMeasurementFormalObservation})
to Conditions \textbf{Ob1}--\textbf{Ob3} in Observement below.

Formally, there is an  \textit{observement system} $(\mathcal{S}, \mathcal{O}, m)$, where the system $\mathcal{S}=\langle S, R\rangle$ consists of a set of objects  $S$ and  a finite set  of relations $R$ on $S$; and the corresponding system $\mathcal{O}=\langle O, P\rangle$  where $O$ is a set  of observations, that is, the set of objects that can be obtained by the observement system, and $P$ is a set of relations on $O$.  
In addition, there exists at least one algorithm $m$ that for any input in $S$ determines an observation in $O$. 

\begin{tcolorbox}
For example, graphs are commonly used to model real-world networks.  Let $S$  be a set of social  networks and let $O$ be  the set of relationships on $S$.   Here the observements of the networks are graphs, and  $O$ is the set of graphs and $P$ is the set of relations on graphs.     
\end{tcolorbox}

The \textbf{observement system} satisfies the following conditions:
\begin{condition}
\textbf{(Ob1) Representation Condition} -- The algorithm $m$  generates a homomorphism $h_{m}:S\rightarrow O$ such that for any relation $r\in R$, there is a corresponding relation $p\in P$ such that for any $x_{1},x_{2},x_{3},\ldots, x_{k}\in S$:
$r(x_{1},x_{2},x_{3},\ldots, x_{k})\Leftrightarrow p(h_{m}(x_{1}), h_{m}(x_{2}), h_{m}(x_{3}), \ldots, h_{m}(x_{k}))$.
\end{condition}

\begin{condition}
\textbf{(Ob2) Existence Condition} -- Clearly, there must be at least one mapping from $S$ to $O$, otherwise there is no algorithm to ``observe'' objects in $S$.   This leads to the \textit{existence condition}: the set $A(S)$ of algorithms that give homomorphisms from $S$ to $O$ is non-empty.
\end{condition}

If the observement system satisfies the following additional condition, we say that the observement system is \textit{strong}.

\begin{condition}
\textbf{(Ob3) Uniqueness Condition} --
The algorithm must be unique in the sense that any two algorithms used to observe the
same object must be related. 
That is, if there exists another algorithm $m'$ giving a
homomorphism $h_{m'}:S\rightarrow O$, then there exists a
mapping $f:m\rightarrow m'$ such that for any relation $r\in R$ and 
the corresponding relation $p\in P$, 
$r(h_{m}(x_{1}), h_{m}(x_{2}),  \ldots, h_{m}(x_{k}))
 \Leftrightarrow$  
$p(f(h_m)(x_{1}),
  f(h_m)(x_{2}),\ldots, 
  f(h_m)(x_{k}))$.  
 \end{condition}

  Condition \textbf{Ob3} requires that the observation obtained by any method of observing a given property can be obtained by a mapping from the result obtained from another   method of   observing the same property. Appropriate standards for these methods of observement should ensure this.  This condition holds for measurements such as  length, distance and mass.  It may also hold for some systems of observement.  
  The observement system is called \textit{weak} if this condition  does not hold.

To elaborate, 
condition \textbf{Ob3} does not  hold for all observement systems.  Unlike measurement, where a mapping can be found between algorithms for measuring with numerical outputs, the case may be more complex for observement. 
As an example of a weak observement system where \textbf{Ob3} does not hold,
consider a simple scale for height: \textit{short}, \textit{medium}, \textit{tall}. Suppose that:
System $A$ defines $small$ as $< 150$ cm and $tall$ as $> 183$ cm; whereas  
System $B$ defines $small$ as $< 155$ cm and $tall$ as $> 178$ cm.
Both systems capture the intuitive and formal properties of height, but there is no mapping between the two.  

The above definition provides a basic set of criteria that can 
serve as guidelines when setting up an observement system. 
We now argue that several such representations already exist. 
Below, we consider two non-numerical categories of observement systems:  the first maps objects to strings, and the second maps objects to graphs. 


\section{Strings} \label{secStrings}

Some experimental methods that produce strings as their outputs satisfy the formal definition of observement. By $strings$  we mean well-formed sequences of symbols within a formal language. 

Data in the form of strings is common because many processes form sequences, especially over time, but also in space or some other ordering. An important case is written language, which consists of sequences of characters arranged according to the relevant syntax. This idea is not restricted to natural languages, but also includes formal languages (e.g. arithmetic) and computer code. For instance, formal languages (e.g. L-system models of plants) have been widely used to describe growth patterns \cite{prusinkiewicz2012algorithmic}. 

Strings have also been used 
to define complexity. 
An information theoretic interpretation is that the complexity 
of a system is the length of a message needed to describe it 
\cite{chaitin1966length, kolmogorov1968three, 
papentin1980order}. 

The use of symbolic strings to represent data is common.  We look at two examples, animal behaviour and genetic codes, that  satisfy the definition of observement.

\subsection{Animal Behaviour}

\subsection{A Simple Example}
Strings have been used to encapsulate sequences of actions. We briefly describe methods for modelling animal behaviour using strings. 

\subsubsection{Turtle Geometry}
A simple string language $S = \langle Alphabet, Syntax\rangle$  consists of all the strings that can be  made from an alphabet by applying its syntax (a set of production rules such as replacement, 
addition and concatenation). 
For example, a grammar for a turtle geometry is given in Figure \ref{box:TG}. 
\footnote{In this and the following examples, we use the conventions of Backus Nauer
Form (BNF) notation to define syntax. 
The symbol `+' denotes 
    one or more repetitions and the symbol `\big{|}' represent alternatives.    
 
}

\setlength{\fboxsep}{10pt}
\begin{figure}[h!]

\begin{center}
\noindent\fbox{%
    \begin{minipage}{0.7\linewidth}
              $S = \langle $Alphabet, Syntax $\rangle$\\
        Alphabet = $\lbrace L, R, F, T \rbrace \cup \lbrace \langle$ path $\rangle\rbrace$\\
        Syntax = $\lbrace \langle$ path $\rangle$, 
       
       $\langle$ path $\rangle \rightarrow$ $F\langle$ path $\rangle$, 
        
        $\langle$ path $\rangle \rightarrow$ $L\langle$ path $\rangle$, 
        
        $\langle$ path $\rangle \rightarrow$ $R\langle$ path $\rangle$,
        
        $\langle$ path $\rangle \rightarrow T $
        
        $\rbrace$
    \end{minipage}}
    \caption{A simple grammar for Turtle Geometry. Here $L$, $R$, $F$, $T$ represent $Left$, $Right$, $Forward step$, $Terminate$ respectively. } \label{box:TG}
    \end{center}
    \end{figure}

The simple grammar in Figure \ref{box:TG} generates strings, 
such as $FFLFFFRFT$, which describe simple paths 
in turtle geometry~\cite{abelson1986turtle}. 
One advantage is that it provides a formal method to compare patterns: similar patterns have similar strings.

\subsubsection{Animal Behaviour}

The idea of representing behaviour as a string of actions has 
 been employed by ethologists \cite{shih2000ethom}. 
An animal's behaviour is recorded as a string of symbols
in a language $L_B=\langle A_B, S_B\rangle$ (Figure \ref{Animal_behaviour}). 
We can regard this system as an observement $\langle S,O,m\rangle$ by defining $S$ to be sequences of animal behaviour and $O = L_B$. 
The mapping $m$ is defined by first assigning symbols to particular actions (the semantics) (Figure \ref{Animal_behaviour}). 
To record a sequence of behaviour, the observer uses an event recorder, 
which is a device with keys
related to predefined actions.
Each time an action occurs, the observer presses the corresponding key
producing a string that describes the animal's behaviour. 
The advantage is that the sequence provides 
an analytic approach that simplifies the task of identifying similar 
or repeating patterns of behaviour. 

The above approach to obtaining behavioural strings ensures that $m$ is a homomorphism. If any sequence of behaviour is followed by any of the defined actions, then the symbol for that action will be the next  entry in the recorded string. So the definition satisfies Condition \textbf{Ob1}. Also, the definition of the recording process ensures 
that at least one algorithm exists, so Condition \textbf{Ob2} is satisfied.
Finally, the systems satisfies Condition \textbf{Ob3} because if we use any other symbols to express individual actions, then simple replacement of 
corresponding symbols defines a mapping from one representation 
system to the other. 

\begin{center}
\begin{figure}[!h]
	\includegraphics[width=\linewidth]{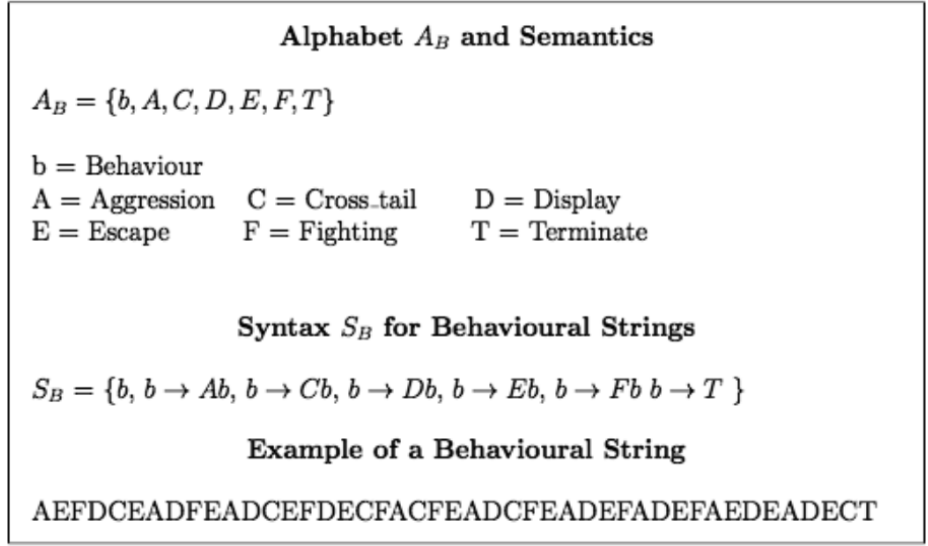}
    \caption{Example of using semantics to record animal behaviour as a string of actions based on \cite{shih2000ethom}.  }
	\label{Animal_behaviour}
    \end{figure}
\end{center}

\subsection{The Genetic Code}
A common representation of data by strings is used to record sequences of genes and proteins. Genetic data consists of DNA sequences where each character is one of four bases:  
Adenine ($a$), Cytosine ($c$), Guanine ($g$) and Thymine ($t$) (see Figure \ref{DNAstring}). 
Most genes code for proteins (strings of amino acids). 
In DNA, the bases are grouped in sequences of three to form \textit{codons}, each of which 
corresponds to an amino acid, or signals the start or end of the sequence
(see Table~\ref{tab:amino}). 

\begin{table}[h]
\includegraphics[width=1\linewidth]{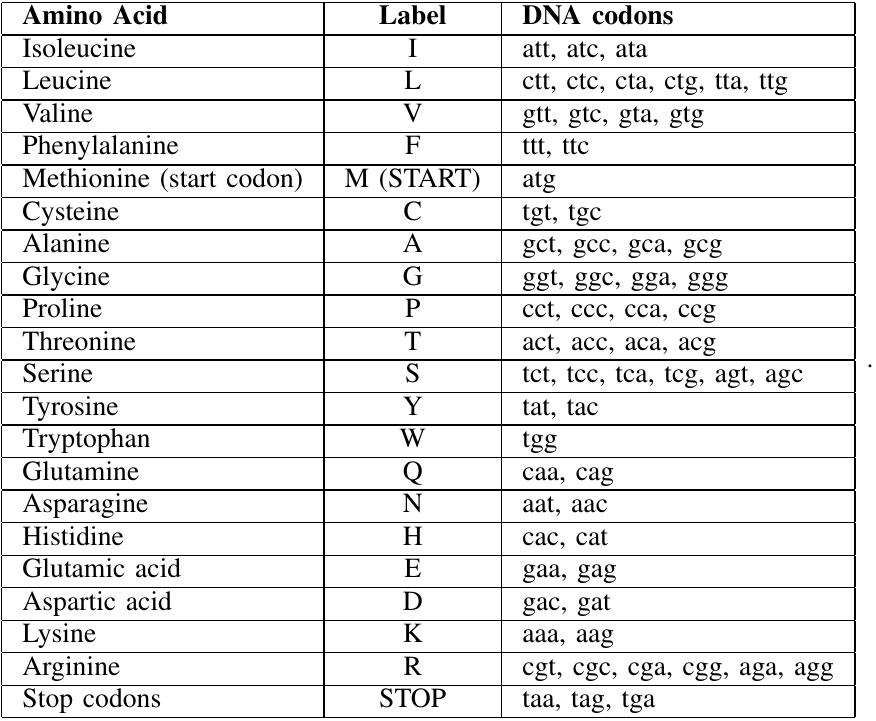}
	\centering
	\caption{DNA codes for amino acids. 
    Note that the start codon also codes for the 
amino acid Methionine, and is usually represented as `M'
    \cite{griffiths2000introduction, graur2000fundamentals}}. 
	\label{tab:amino}
\end{table}

For instance, the DNA codon $atg$ signals both the start of a gene sequence 
and the amino acid Methionine ($M$), the codon $atc$ produces Isoleucine ($I$), and the 
codon $tag$ signals the end of a gene. There are 64 codons, but only 20 amino acids, 
so there is a lot of redundancy. For instance,  
there are six codons 
$tct$, $tcc$, $tcs$, $tcg$, $agt$, and $agc$ 
for the amino acid Serine ($S$). 
Thus there is a strict correspondence between the DNA sequences in genes and the 
amino acid sequences in proteins (Figure \ref{DNAstring}).   

\begin{figure}[!h]
\noindent\fbox{%
    \begin{minipage}{\linewidth}
\begin{center}
	\includegraphics[width=\linewidth]{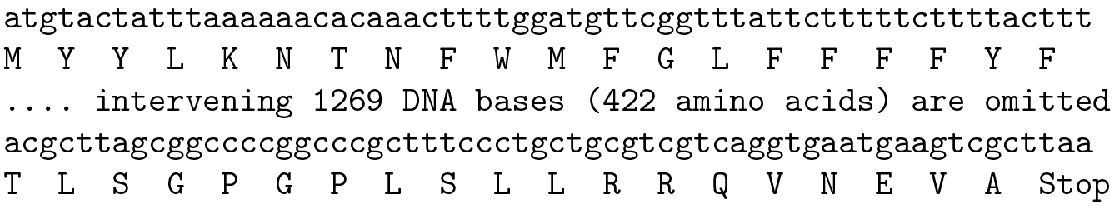}
\end{center}
\end{minipage}}
	\caption{The start and end of a DNA sequence of 1554 bases and the corresponding amino sequence for lactose permease [$Escherichia$ $coli$] (\textit{Source: GenBank ID AAA24054.1}).  Each subsequence of three bases (a \emph{codon}) 
codes for an amino acid. The amino acid string corresponding to this gene 
is listed beneath it     
    \cite{Genbank} }
	\label{DNAstring}
    \end{figure}

\begin{figure}[!h]
\begin{center}
	\includegraphics[width=0.9\linewidth]{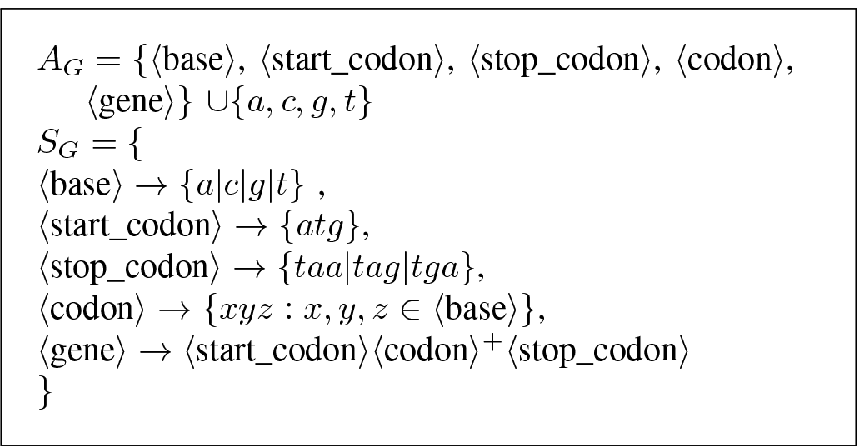}\end{center}
    \caption{Language for describing gene sequences. We use the conventions of 
    Backus Nauer form (BNF) notation to define syntax. The symbol `+' denotes 
    one or more repetitions and the symbol `\big{|}' represent alternatives.    
    }\label{figGS}
    \end{figure}

To interpret the acquisition of gene sequences as observement, we define 
a  language $L_G = \langle A_G, S_G\rangle$ to describe gene sequences 
(Figure \ref{figGS}). 
So for genes, we define the observement system $\langle S,O,m\rangle$ by 
interpreting $S$ to be the set of all genes and setting $O$ to be the language $L_G$. 

To define the mapping $m:S \rightarrow O$, we associate the constants ($a, c, g, t$) in $A_G$ with the DNA bases listed above. 
A comprehensive methodology exists for recording gene sequences including  processes to identify the base sequences and a variety of software tools to construct sequences and identify the genes ~\cite{Gilbert1974Sequences}, ~\cite{Genbank}. 

The above definition satisfies Condition \textbf{Ob1}. The mapping $m$ is a homomorphism, because both real genes and the sequences that describe them are strings, and $m$ satisfies the same rules for strings: any substring of a sequence corresponds to a section of the gene it represents. 
The above definitions also ensure that there is at least one algorithm 
to obtain a description of a gene sequence, so it satisfies condition \textbf{Ob2}.
Finally, gene sequencing satisfies the \textit{strong} Condition \textbf{Ob3}.  Suppose, we used a different mapping $m: S\rightarrow O'$ where $O'$ were different symbols to represent the four bases.  
Then the mapping $f: O'\rightarrow O$ which performs a simple replacement of corresponding symbols gives a mapping from one  system to the other. 

Although the observement of genes described here is simple, the biological processes involved in the translation of genes into proteins are complex 
Many biological issues  are beyond the scope of this 
discussion, such as reading frames, exons and introns, controller genes, and 
the roles of messenger RNA and ribosomes.

\subsection{Proteins}

Most genes provide the code for producing proteins, which are formed as strings 
of amino acids. 

Just as we did for genes, we can define a simple language $L_P = \langle A_P, S_P\rangle$ 
to represent proteins as amino acid sequences (Figure \ref{fig:proteinExample}).

\begin{figure}[!h]
\begin{center}
\includegraphics[width=0.9\linewidth]{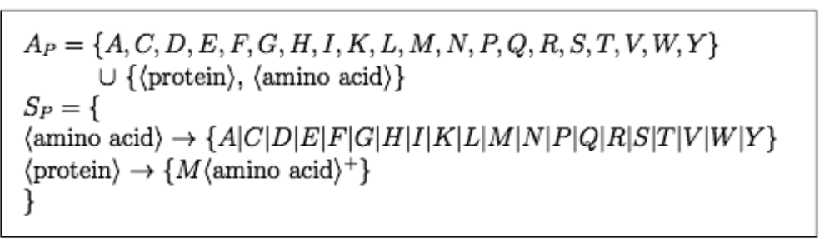}
\end{center}
\caption{Language for defining proteins as strings of amino acids.}\label{fig:proteinExample}
\end{figure}

This language will produce any amino acid sequence (e.g. Figure \ref{DNAstring}).

The proof that amino acid sequencing is a \textit{strong} observement system parallels 
the argument for gene sequencing almost exactly. The standards associate the 
constant symbols with the amino acids they represent (\textbf{Ob1}). As for DNA, there are 
many widely known methods (\textbf{Ob2}) for extracting and interpreting amino acid strings 
and protein structure \cite{Genbank}. 
As Figure~\ref{DNAstring} shows, the observement systems for representing genes 
and proteins are related. 
Gene sequences map to protein sequences (\textbf{Ob3}). 
The equivalences (Table~\ref{tab:amino}) define a 
homomorphism from gene sequences onto amino acid sequences.

\section{Graphs}\label{secGraphs}

\subsection{Network Representations}\label{secNetworks}

Networks have gained increasing prominence in many areas \cite{Gysi2020-RoyalSociety}. Examples include social networks,  infrastructure networks, and software systems.  
They are used in the analysis of biological
networks, such as food webs and genetic regulatory networks.
Diagrams of networks are widely used to convey
information, such as organisational structure, family trees,
flow diagrams and semantic relationships 
(see Figure \ref{Graph_data}). 

As networks underlie many diverse fields and applications, it is important to have  methods of observing behaviour in networks, understanding interactions between entities within networks, and comparisons between networks.

Graphs are widely used to model networks and network behaviour.  
A graph is a pair $(V,E)$,  where  $V$ is a set of nodes (or vertices) and  $E\subseteq V\times V$ 
is a set of edges (or arcs) connecting pairs of nodes.  (A graph with directed edges is called a \textit{digraph}.)  
For clarity, here we define networks to be graphs in which the nodes and edges can have associated attributes.    
Nodes are an abstraction of entities in the networks, and  edges are an abstraction of relationships between entities.  For example, the graph in Figure \ref{Graph_data} models species as nodes and the `to eat' relationship by directed edges.

Many important  network properties can be determined from a graph.  The use of graphs to  model networks is powerful, as it allows the use of a wide range of well-developed tools and methods to extract valuable information about networks such as network reliability \cite{EGL-1991, acgbn-2015}, connectedness \cite{p-1980, sb-2002}, common network structure  \cite{ESDS-2016, pksbk-2012}, efficient resource allocation \cite{pzejp-2014, EAPPGSKMM-2011}, to identify flows within networks \cite{YL-2013,umts-2009}, and efficient route detection \cite{cdpp-2012, pk-2009}.  Graphs and networks provide a common theoretical model for patterns of interactions, where common interactions inside a network can be represented as a subgraph (subnetwork).

For example,  Figure \ref{Graph_data}(a) is a directed graph where each node represents a species (frogs, spiders, insects, etc) and each directed edge represents a relationship between species (e.g. $(frog,insect)$ denotes `frogs eat insects").


\begin{figure}[htb!]
	\centering	
  \includegraphics[width=0.95\linewidth]{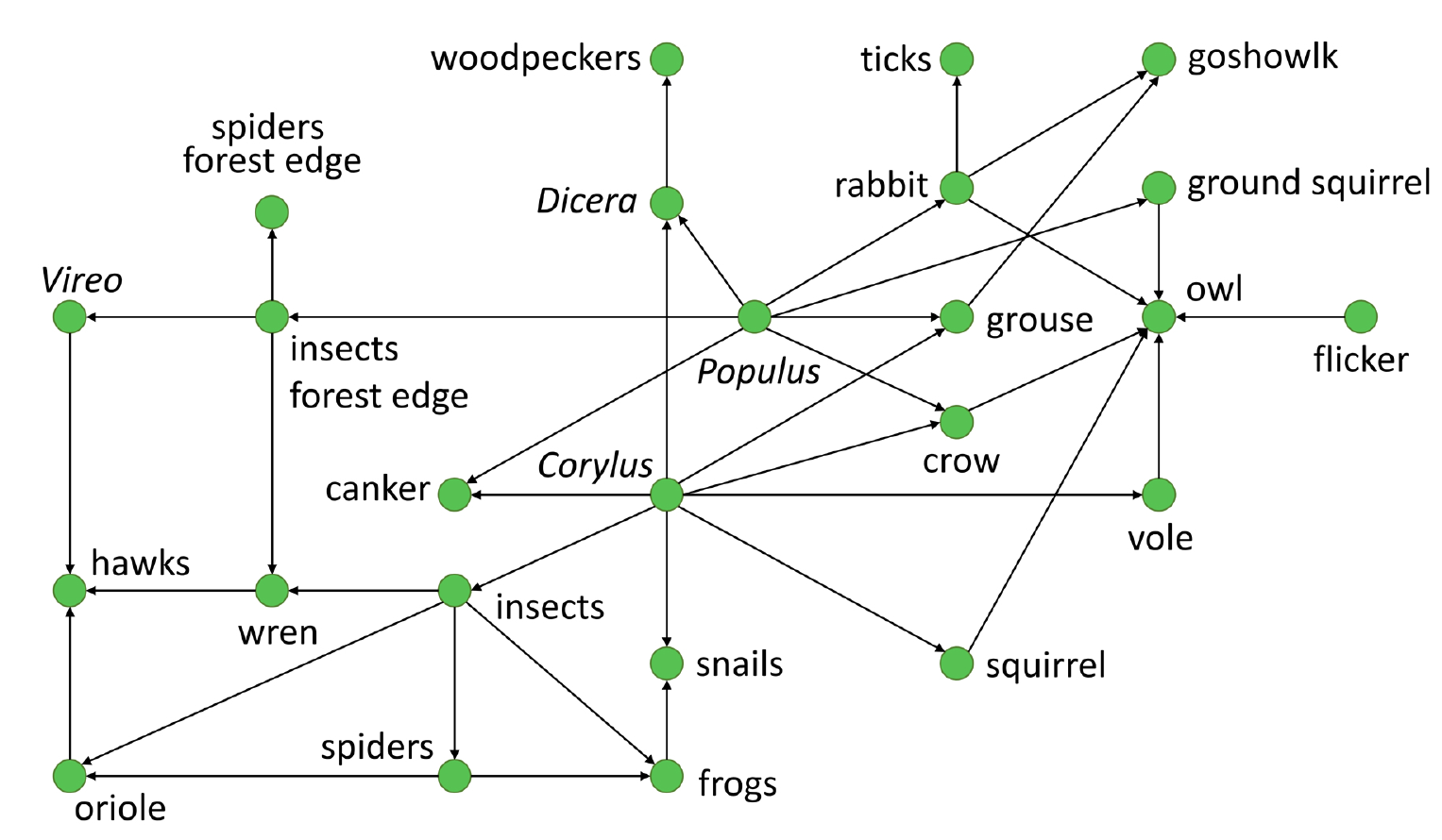}\\
   (a)\\
   
 \vspace{5mm}
 \includegraphics[width=0.8\linewidth]{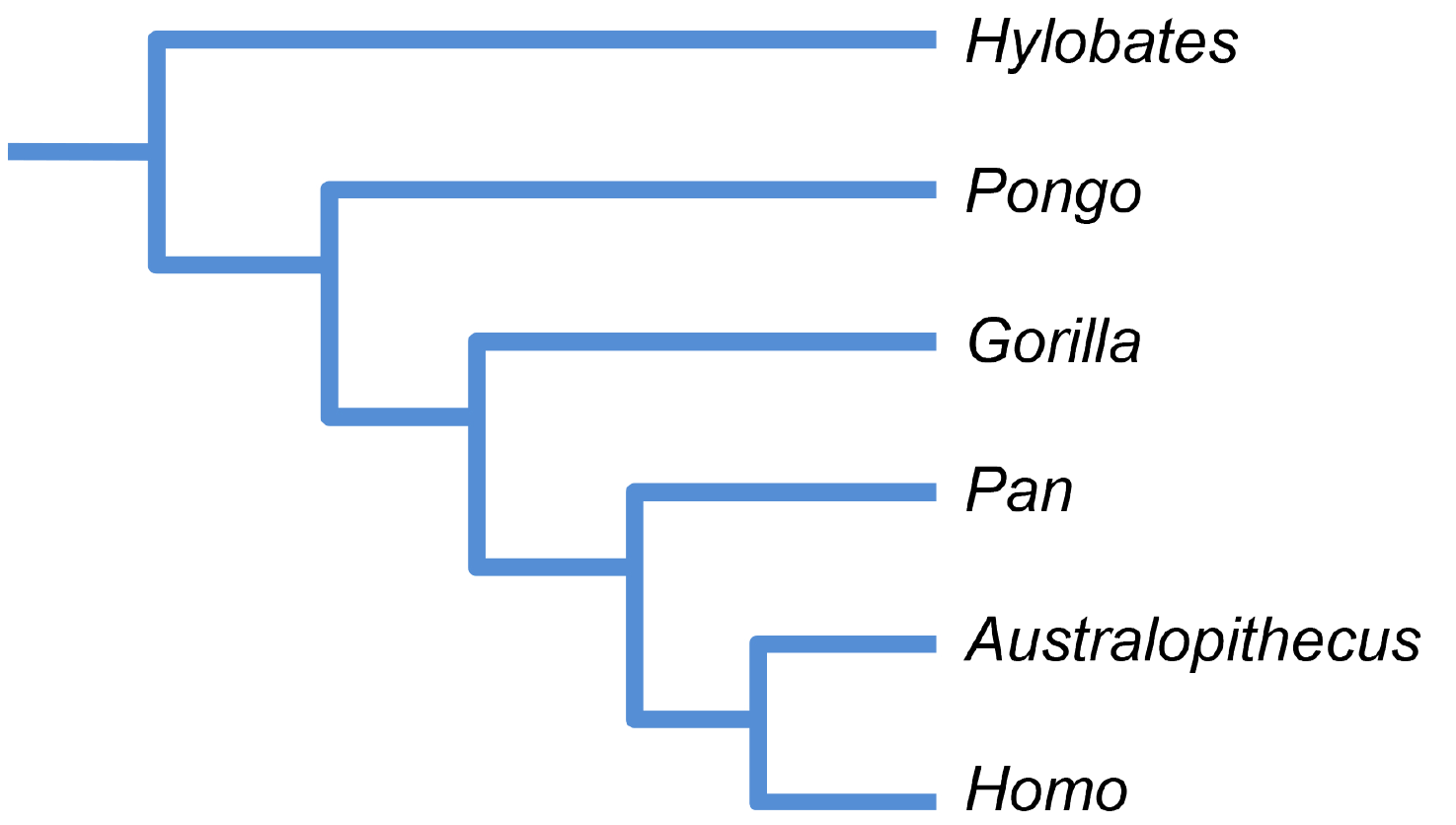}\\
 (b)\\
	\caption{Examples of data represented as networks (graphs with values for nodes/edges). 
    (a) A food web from Aspen Manitoba (redrawn from data in \cite{cohen2012community}). 
The arrows denote that one species serves as food for another. 
    (b) A phylogenetic tree (dendogram), showing inter-specific relationships (after 
    \cite{strait2015analyzing}). 
       }
	\label{Graph_data}
\end{figure}

There are existing standards for  representation of graphs as networks.  
Common data structures for storing graph data include \textit{adjacency lists} and \textit{adjacency matrices} (see Figure \ref{fig:graph1}).  These structures can also be represented  as a bit string such as the
\textit{graph6} format \cite{GRAPH6}. 

\begin{figure}[!h]
\begin{center}
  \includegraphics[width=0.75\linewidth]{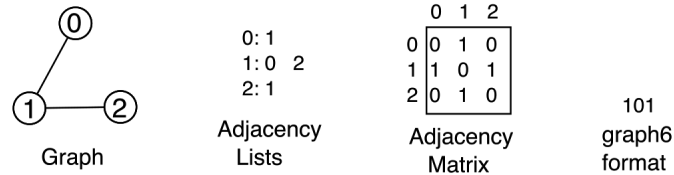}
  \end{center}
  \caption{Different Graph Representations - a diagram, adjacency list, adjacency matrix, and bit-string representation of the upper triangular adjacency matrix (\textit{graph6} format).}
  \label{fig:graph1}
\end{figure}

 
A proof of universality of graphs in underlying the structure of all complex systems is given in \cite{Green1992Emergent}.   
This universality of graphs means that the network models provide
important insights about many different systems 
\cite{paperin2011dual}. 
Certain network topologies such as trees occur widely and 
are known to convey important properties and behaviours. 
Perhaps the most far-reaching insight was the proof that random graphs undergo a critical 
phase change, from fragmented to connected, as the edge density increases \cite{er-1960}. 
This property of graphs accounts for a wide range 
of physical phenomena \cite{Green1994Connectivity}, such as crystallization, firing of a laser and the onset of percolation.  

\subsubsection{Observement System for Graphs} 
In this section, we demonstrate that graphs and related methods, data structures and standards are an observement system for networks.
Here graphs are viewed as observations of networks, and the mapping of the network to a graph data structure   is a homomorphism between common graph relations and network relations.    For example,  the subgraph relation is homomorphic to the subnetwork relation, and the isomorphism relation is homomorphic to `equality'  between networks. 

We  demonstrate that a mapping $m$ exists that gives an homomorphism between relations on the network and relations on the graphs and that  satisfies the observement properties. 

\subsubsection*{A Graph Observement System}
Formally, we can regard the system $\langle \mathcal{S}, \mathcal{O}, m\rangle$ of observements of networks as graphs as follows.  

Let $\mathcal{S}=(S,R)$ where $S$ is the set of networks of interest  and  $R$ is the set of relations on these networks.  An example of a network relation is the \textit{subnetwork} relation, where network $N_{1}$ is related to network $N_{2}$ if and only if 
network $N_{1}$ is a subnetwork of network $N_{2}$.  

Let $O$ be the set of graphs (observations) and $P$ be the set of relations on these graphs.  An  example of the graph relation  is the \textit{subgraph} relation  where $G_{1}$ is related to $G_{2}$ if and only if $G_{1}$ is a subgraph of $G_{2}.$

There are many possible mappings from real-world networks to graphs.  Here we use an adjacency matrix (e.g. per Figure \ref{fig:graph1}) where the entities in the network are represented as vertices and the relationships between pairs of entities as edges.  This is a well-established structure for representing graphs and has associated methods and standards. 
The mapping $m$ first creates an $n\times n$ adjacency matrix   where $n$ is the number of vertices. It then assigns non-zero entries corresponding to relationships between pairs of entities.  Thus, $m$ maps a network to an observation 
(graph).

\subsubsection*{Ob 1 is satisfied} 
This system satisfies Condition \textbf{Ob1} as the algorithm $m$ is a homomorphism from networks to graphs that maps relations on networks to relations on graphs.  For example, the mapping of the subnetwork relationship to the subgraph relationship.   
An example of the usefulness of these types of relation-preserving mappings  
is the mapping of subnetworks with significant interactions in biological networks  to subgraphs (network motifs) in graphs.  

\subsubsection*{Ob 2 is satisfied}
There is at least one algorithm $m$ that can be  used  to represent a network by an adjacency matrix  Thus,  the system $\langle \mathcal{S}, \mathcal{O}, m\rangle$ satisfies Condition \textbf{Ob2}.

\subsubsection*{Is Ob 3 satisfied?}
Clearly, there is a mapping between any of the common graphs representations (adjacency matrices and adjacency lists, graph drawings, and compact storage encodings e.g. graph6 \cite{GRAPH6}) and so   \textbf{Ob3} is  satisfied for these representations.
   Although, there exist mappings between algorithms that observe networks as graphs represented in these formats, it is an open question where this is always the case.

\subsection{Graph Observement Systems - Applications}

Networks have been employed to analyze a wide variety of data
from medieval politics \cite{yose2018network} 
to geospatial patterns \cite{su2016data}. 

A powerful relation that is preserved when mapping networks to graphs is the mapping of the subnetwork relation on networks to the subgraph relation on graphs.  Identifying common substructures in networks/graphs underlies many research areas; as ``\ldots [g]raphs seem to be the current answer to
the question no matter the type of information: molecular data, brain images or neural signals'' \cite{TDL-2016}. 
Thomas, Dongmin and Lee's survey of
similarity relations on graphs 
\cite{TDL-2016}, in particular, shows
how they can be used to understand neurological interactions and to identify neurological disorders.  Another example \cite{sasd-2015}, uses network topologies based on interactions between and within subnetworks to investigate  changes in the brain in people with  Alzheimer's disease.  Although,  their measures of connectivity are numbers such as path lengths and clustering coefficients, the comparison of the networks represented in this way adds a meaningful layer to the studies (see Figures 4 and 5 in \cite{sasd-2015}).

Trees are a common way to represent hierarchical relationships in many fields.  For example,
dendograms are standard tools for representing community structure in large networks \cite{PhysRevE.70.066111} and  are widely used in representing taxonomic relationships (Figure \ref{Graph_data}(b)).  


\subsection{Family Trees}
Genealogical information is often represented by  a directed acyclic graph (DAG), commonly known as a \textit{family tree}.  In this representation, the family members are represented by nodes and the ``child of" relationship is represented by a directed edge from parent to child and the "partnered with" relationship is represented by a bi-directed (or undirected)  edge.  The underlying graph representation is technically not a tree, as there may be more than one path between a pair of nodes.  A graph with directed edges is called a \textit{digraph}.

Large databases of genealogical data are maintained, for example \textit{ancestry.com}.  An interesting account of different visualisations of this data is given in \cite{U_VFT}.

A digraph $D$ can be represented as a pair $D=(V,E)$ where   $V=\lbrace 0, 1, \ldots, n-1\rbrace$ is the set of nodes, or vertices, and $E\subseteq V\times V$ is the set of directed edges connecting pairs of nodes.  Nodes may be labelled with information such as names and  date and place of birth. 

The empty digraph denoted $D(\emptyset)$ has a single vertex and no edges.  We use the notation $D_{1}+_{\overrightarrow{u_{1}u_{2}}}D_{2}$ and $D_{1}+_{u_{1}u_{2}}D_{2}$ to denote the digraphs obtained by connecting graphs $D_{1}$ and $D_{2}$ by the directed, or  bidirectional, edge  $u_{1}u_{2}$ respectively where $u_{i}\in V(D_{i})$ for $i\in \lbrace 1,2\rbrace$.    Similarly, we denote adding a directed, or bidirectional, edge between vertices $u$ and $v$ in $D$ by  $D+\overrightarrow{uv}$  or $D+uv$ respectively.  

If some relationships are missing, we may have a disconnected family tree.  In such cases the $D$ is a disjoint union of digraphs $D_{1}, \ldots, D_{k}$ which we denote as $D=D_{1}+\ldots +D_{k}$.

A digraph $D$ can then be recursively defined as:
\begin{align*}
  D:=& D(\emptyset)\text{  } | \text{  }D_{1}+_{\overrightarrow{u_{1}u_{2}}}D_{2}\text{  } | \text{  }D_{1}+_{u_{1}u_{2}}D_{2}\text{  }\\
  &\text{  }|\text{  } D+\overrightarrow{uv}\text{  }|\text{  }D+uv \text{  }|\text{  }D_{1}+D_{2}.  
\end{align*}

The following relationships are encapsulated within the digraph:
\begin{itemize}
\item\textbf{is child of}: For all $u,v\in V(D)$, $u$ is a child of $v$ if and only if $\overrightarrow{vu}\in E(D)$.  
\item\textbf{is parent of}: For all $u,v\in V(D)$, $u$ is a parent of $v$ if and only if $\overrightarrow{uv}\in E(D)$.  
\item \textbf{partnered}: For all $u,v\in V(D)$, $u$ partnered $v$  if and only if $uv$ is a bidirectional edge in $E(D)$.
\item \textbf{is related to}: For all $u,v\in V(D)$, $u$ is related to $v$ if and only if $u$ and $v$ are  connected by a  path (disregarding the direction of the edges) in  $D$.  
\item \textbf{is descendant of}: For all $u,v\in V(D)$, $u$ is a descendant of $v$ if and only if there exists a directed path $u=u_{0}, u_{1}, u_{2}, \ldots, u_{k}=v$ in $D$ where  each edge $\overrightarrow{u_{i}u_{i+1}}\in E(D)$ for $i\in \lbrace 0, \ldots, k-1\rbrace$.
\item \textbf{is predecessor of}: For all $u,v\in V(D)$, $u$ is a predecessor of $v$ if and only if there exists a directed path $v=v_{0}, v_{1}, v_{2}, \ldots, v_{k}=u$ in $D$ where  each edge $\overrightarrow{v_{i}v_{i+1}}\in E(D)$ for $i\in \lbrace 0, \ldots, k-1\rbrace$.
\end{itemize}

The underlying data representation of a digraph is commonly an adjacency matrix or adjacency list representation.  Clearly, there is a mapping between these two representations (and so the Observement Condition Ob3 is satisfied for these representations).  
To reflect its inherent symmetry, the
`partnered with' relationship is represented by a bi-directed   edge.  The underlying graph representation is technically not a tree, as there may be more than one path between a pair of nodes.

Due to  their hierarchical nature, a common way to visualise genealogical trees is as an \textit{ancestry chart} (for example, \cite{U_VFT}) where a selected node is positioned as the root of the `tree' (or subtree) and a tree-like structure emanates from the root. The problem of representation and layout of family trees is a specialisation of general graph layout problems (for example, \cite{s-2011}).  Family trees can be extended to other types of  genealogical networks (for example, \cite{GK-2007}).

As the size of the tree increases, the tree-like layout maybe replaced by alternative visualisations.   Systems for extracting genealogical data and visualising the genealogical structure using matrix representations are given in  \cite{bdfbw-2010}.

\subsection{Dendograms and Motifs in Graphs}

A \textit{dendogram} is a (tree) graph that illustrates the hierarchical relationships between clusters of data 
 (see Figure \ref{Graph_data}).  
Each leaf node corresponds to a particular cluster, and clusters corresponding to leaf nodes in the same sub-tree have common properties.  Leaf nodes belonging to smaller sub-trees are more closely related than those connected only by larger subtrees.  Dendograms are standard tools for representing community structure in large networks \cite{PhysRevE.70.066111}. 

Just as for strings, motifs are also used in the interpretation of graphs (see Figure \ref{Motifs}). 
Network \textit{motifs} are subgraphs that appear to occur more frequently than expected in certain networks \cite{M2002}. Often the motifs of a network are related to meaningful interactions such as social interaction \cite{G2002}, protein interactions and  gene transcription \cite{Y2004, burda2011motifs} and other biological interactions \cite{M2002, benson2016higher}.


As we found earlier with strings,  motifs can be used to  represent common patterns between graphs which in turn represent networks.  However, at another level, motifs can be used as an observement themselves.   Thus, we can have a  mapping from  $N\rightarrow G \rightarrow M $ which maps a network to a graph, and a graph to a motif. 

\subsection{Graphs as theoretical tools}

Graphs and networks have gained widespread theoretical prominence in the study of complex systems. This is because they provide a common theoretical model for patterns of interactions. The following theorems guarantee this by providing a proof of universality
\cite{Green1992Emergent, Green1994Connectivity}.  

\begin{theorem}
Graphs underlie the structure of all complex systems.
\end{theorem}
\begin{proof}
  The proof \cite{Green1992Emergent, Green1994Connectivity} rests on the observation that models of complex systems use only a small number of representations (e.g. matrices, systems of equations, cellular automata). So if we show that graphs are implicit in those representations, then graphs are implicit in the structure of any complex system for which that representation is used to create a valid model. 
\end{proof} 

\begin{theorem}
  In any deterministic automaton with a finite number of states, the state space forms a directed graph.
\end{theorem}
\begin{proof}
  For any automaton $\langle A,S\rangle$, the proof \cite{Green1992Emergent, Green1994Connectivity} defines the set of states $A$ 
to be nodes in a graph and the successor function $S$ defines a set of edges $R_{A}$ where 
$R_{A} = \lbrace (x,y) | S(x) = y\rbrace$,  where $x,y  \in A$. 
\end{proof}

The proof extends to arrays of automata, since in any such system, the 
suite of individual states at any stage defines a state for the array. 
By defining edges for any transition with non-zero probability, the 
result also extends to stochastic processes.

This universality of graphs means that the network model provides
important insights about many different systems. 
Certain network topologies (trees, for instance) occur widely and 
are known to convey important properties and behaviours. 
Perhaps the most far-reaching insight was the proof that random graphs 
undergo a critical 
phase change, from fragmented to connected, as the edge density increases \cite{er-1960, Green1994Connectivity}. 
This property of graphs accounts for a wide range 
of physical phenomena, such as crystallization, firing of a laser and the onset of percolation.

\subsection{Other Applications}

Graphs have proved a convenient way to represent transportation, utilities and other large-scale infrastructure. 
A useful abstraction is ``networks of networks'' \cite{d2014networks}, which allows efficient 
identification of key properties such as network reliability (important in utility networks), network flow (important in transportation networks), and catastrophic cascading failures
\cite{brummitt2015coupled}. 
These critical properties are preserved by the  homomorphism (Condition \textbf{Ob1}) illustrating the benefits of observement systems: they abstract and simplify while maintaining these important relationships.

Another example of the usefulness of graphs in  observement is modelling the spread of disease by encapsulating the connections between members in the community (e.g. \cite{GB-2007}, \cite{agh-2011}).  These observations (graphs) can be used to identify how to fragment the graph into subgraphs that minimise interactions and thus minimise the spread of such diseases.


\section{Applications of Observement Theory}

The generalisation of the concept of measurement to include non-numerical data 
prompts us to ask what kinds of analytic methods can be used. Numerical
methods may no longer necessarily apply, so there is a need to identify new kinds of 
analyses that apply to each kind of representation. 



\subsection{Motifs}

Motifs are small, re-occurring structures with significant meaning. They occur in many contexts. In music, for instance, a short phrase (\textit{leitmotif}), is sometimes used to identify people or characters 
in a ballet, opera or movie. In recent years they have become 
popular analytic tools and are used to analyze both strings and graphs. 

In strings, motifs are substrings that  typically occur more frequently.
A good example of motif use is in the observement of proteins, via strings of amino acid sequences~\cite{Genbank}. 

Many methodologies have been developed for interpreting gene 
and amino acid sequences~\cite{Genbank}. 
Most relevant here are methods that compare different sequences. 
For genes, important questions include  relationships between 
organisms, and unravelling genetic regulatory networks. 
For proteins, some of the most important questions concern 
their structure and function. 
Motifs can be used to describe families of proteins and interrogate their structure. 

A good example is the use of motifs to describe families of proteins, 
as pioneered by Bairoch and his colleagues. 
The protein database Prosite 
\cite{bairoch1991prosite} stores information about the structural 
homology of families of proteins \cite{hofmann1999prosite}. 
The database represents proteins as amino acid sequences and 
uses motifs to capture homology patterns. 

\textit{Motifs} are short sections of strings that are associated with  
known structural features, such as alpha helices or beta sheets. 
They are characterized by identifying:
\begin{itemize}
\item 
sets of amino acids that share common physical properties~\cite{chen2002classification}; and 
\item
sequences that play a role in folding. 
\end{itemize}

Some typical rules are given in Figure \ref{Rules}.

 \begin{figure}
\setlength{\fboxsep}{10pt}
\begin{center}
\noindent\fbox{%
    \begin{minipage}{0.95\linewidth}
 \includegraphics[width=0.95\linewidth]{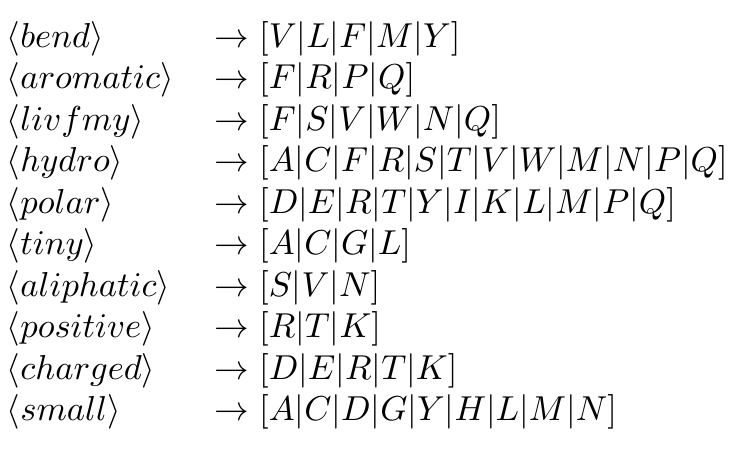}
\end{minipage}}
\end{center}
\caption{Rules for construction of protein representation.}\label{Rules}
 \end{figure}
 
In practice, the database does not use these labels in its descriptors, preferring 
instead to provide a standalone code. 
Figure \ref{Motif} gives an example of three amino acid sequences and the motif pattern they share in common. 
This pattern highlights the structural similarities between the sequences. 
In this example $x(N)$ means ``any $N$ amino acids'' and $\lbrace A,B,C\rbrace$ means 
any one of the listed amino acids.

\begin{figure}
\noindent\fbox{
   \begin{minipage}{0.9\linewidth}
   \includegraphics[width=3.0in]{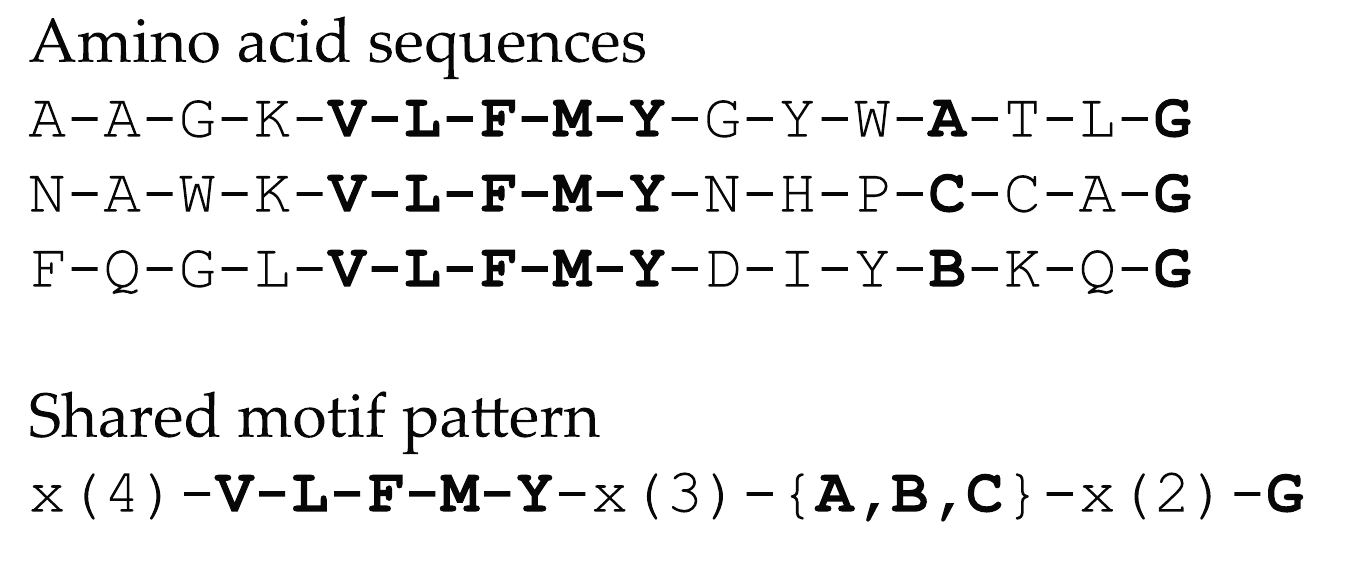}
   \end{minipage}}
\caption{Using motifs to express similarities between amino acid sequences.}\label{Motif}
 \end{figure}

In a similar vein,
gene sequence analysis uses alignment methods to identify 
similarities and differences between corresponding genes of distinct species (for example, see \cite{LHWFRHMAD_2009}); 
several methods make use of motifs to assist in this process.  


As we saw earlier with strings, motifs can be used to represent common patterns between graphs, which in turn represent networks (see Figure \ref{Motifs}). 
\textit{Network motifs} are small sub-graphs that occur more frequently in particular networks of interest 
 \cite{M2002}. These motifs can be related to meaningful interactions such as social interactions \cite{G2002}, protein interactions and  gene transcription \cite{Y2004, burda2011motifs} and other biological interactions \cite{M2002, benson2016higher}.  

\begin{figure}[!h]
	\centering	
 \includegraphics[width=0.8\linewidth]{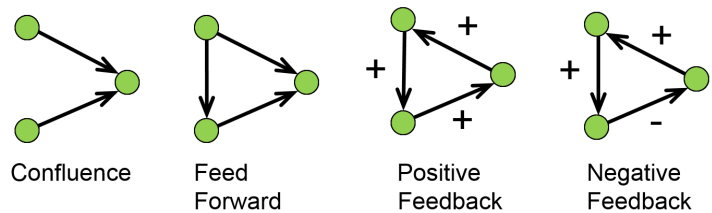}
	\caption{Simple examples of network motifs. 
       } 
	\label{Motifs}
\end{figure}

\subsection{Templates}

In the structure of complex systems, \textit{motifs} are closely related to 
\textit{modules}, which are self-contained, repeatable units of structure.
Examples of modules include large corporations, which divide their operations into self-contained 
units, such as finance, manufacturing; and plants, which grow by repeatedly adding modules, 
such as branches, leaves and buds. 

Shared motif patterns in families of proteins, such as in Figure \ref{Motif}, 
leads to the idea of templates that describe a common pattern shared 
different phenomena share. 
The idea of templates and modules is familiar in language.
For instance, the sentences  

``A cat ate my canary.''	and	``The dog buried a bone.''

are analogous in the sense that share the common grammatical structure

$<article> <noun> <verb> <object>$.

In this case, the underlying template 
relates to describing actions in the real world. 

The idea of templates is also  widely used in taxonomy. 
For example, the common body plan for arthropods 
(insects, spiders, crabs etc)  is 

$<HEAD><SEGMENT>^+<TAIL>$

where $<SEGMENT>$ is a body unit with two legs.

\subsection{Complexity}

An important property of measurement is that it can make vague concepts precise \footnote{A good example is the way measurement of diversity transformed ecology. Initially, the notion of diversity was vague, but in the mid 20th Century, ecologists introduced a succession of metrics \cite{simpson1949measurement, pielou1969diversity}.}. 
However there are many concepts that numerical measurement cannot represent adequately;
such as \textit{organisation} and \textit{process}\footnote{Also, particularly in the humanities, certain observements involve entangled contexts (such as historical and social conventions, see Section \ref{secGeneralisationMeasurement}).}.
As we saw earlier (Section \ref{secNetworks}), graphs are implicit 
in all complex systems. 
Graphs-as-observements handle this problem gracefully, as they can do double-duty: such as in database Entity-Relationship Diagrams (ERDs) for the former; and flowcharts for the latter.
Recalling Section \ref{secGeneralisationMeasurement}, this means that many systems can be represented as networks of nodes and edges, 
which is more intuitive than reduction to numbers, as is the standard conception of conventional measurement \cite{finkelstein1984review,michell2007representational}. 

Several systems base measures of complexity on strings.  
They define the complexity of a system is the length of the 
shortest message required to describe the system
\cite{chaitin1966length, kolmogorov1968three, solomonoff1964formal}. 
Wallace refined this idea, basing his idea of Minimum Message Length 
on the computational model of a $program$ plus $data$
\cite{wallace2005statistical}. 
This idea is consistent with Papentin's division of complexity 
into two components: 
\textit{primary order} (ordered complexity or pattern) and 
\textit{secondary order} (entropy, the random complexity) 
\cite{papentin1980order}.  
All of these approaches rest on the assumption that there is an 
absolute minimum value. 

An alternative is to define complexity 
$relative$ to a particular observement frame (based either on a graph or 
string representation)
\cite{green1996towards, green2011elements, green2001towards}. 
This is a more practical approach to dealing with 
complexity because we can apply a consistent method (i.e. the same 
frame of reference) when comparing different systems. 
For instance, a social network of key power-brokers in an organisation would look very different to the network of social sports participation of employees within the same organisation. 

Conversely, any observement frame, based on graphs or strings, immediately
assigns a complexity value to the observed object. 
We can represent any complex systems as a graph, and 
we can describe that graph as a text string (see Figure \ref{fig:graph1}). 
The length of that text string then provides a number, which is a measure
of the system's complexity. 
That is, there are mappings:
$$GRAPH \rightarrow STRING \rightarrow NUMBER$$
Thus we can associate a measure of relative complexity with any observement based on graphs, or strings. 

\subsection{Data storage and compression}
\subsubsection{String-based Methodologies} 

As we saw in Figure \ref{Motif}, protein analysis uses motif patterns to infer features of protein structure. 
Such ideas are also common in syntactic pattern recognition
\cite{brayer1982syntactic, schalkoff1992pattern}, 
which applies parsing and other inference methods to formal descriptors of observed patterns. 
Here,  motifs can be viewed as representing common patterns between strings.  However, at another level, motifs can be used as an observation themselves.  

Thus, we can have a nested `Russian \textit{Matryoshka} doll' mapping
$P\rightarrow S \rightarrow M \rightarrow  \mathbf{N}$
which maps a protein to a string, a string to a motif, and a motif to a number. Each of these is an observation in some 
observement system,
but  at each level of the nesting, there is some loss of information. Observe that the reduction to numbers, being the ultimate step of the mapping, resulted in the loss of more information.

Several methods exploit the concept of string motifs by using
syntactic and related approaches to storing data. 
A simple example is LZW compression, which maps data to a string \cite{ziv1977universal}. 
It achieves compression by creating a dictionary of repeating 
elements in the string. 
This method achieves lossless compression in certain kinds of data, 
such as large-scale genomic databases \cite{kuruppu2010relative}.
LZW compression could be regarded as an observement that maps data 
to a pair $\langle M,S \rangle$, where $M$ is a set of string motifs 
and $S$ is a string in which motifs are replaced by identifying codes. 

\subsubsection{Graph-based Methodologies} 
\label{secDatabases}
There are numerous examples of databases that demonstrate the usefulness of observements in storing data.


\textit{Graph databases} highlight a shift from the emphasis on collecting data in rigidly-defined and \textit{normalised} tables \cite{codd1970} to collecting information about relationships.
They are designed ``to store data about relationship-rich environments'' \cite{coronel2016}. Graph databases are particularly useful where the interactions between entities are important, or where the shape of the structure is meaningful, or when no suitable pre-determined `template' exists to represent (often incomplete) data. Hence, its popularity lies in its innate suitability for social network applications, where ``relationships become just as important as the data itself'' \cite{coronel2016}.
 

A survey of these databases is given in \cite{AG-2008}, which lists advantages of using this representation;
e.g. users are able to easily visualise their data and relationships;
queries are easily associated with familiar graph operations;
and the availability of special structures to store graphs and efficient algorithms to
process them. 

The implicit graphical nature of information has also led to new kinds of analytic methods for knowledge discovery in large databases. 
A widely used example is \textit{association analysis}, 
which searches for connections between items in a database
\cite{agrawal1993mining,hipp2000algorithms}. 



\section{Conclusion}

In this study we have shown that the formal definition of measurement, notwithstanding the requirement of numeric representation, extends to other, non-numeric representations. We call such systems \textit{observement}, as a generalisation of, and to distinguish them from, traditional numeric measurement systems. We have also shown that several systems based on two commonly used data representations - strings 
(Section \ref{secStrings}) and graphs (Section \ref{secGraphs}) - immediately satisfy the definition. 

Revisiting the crucial roles of a traditional measurement system, listed in Section \ref{secIntro}, observement allows the gathering of data in a standard way \textit{qua} measurement.
Standard representations lend themselves to standard methods of
interpretation. 
As we have seen, strings and graphs are already used to represent many different kinds of data;
and general methods, such as motif detection and 
association analysis can be used in many different areas of study. 

Secondly, observement produces data with well-known properties. 
For example, strings and graphs share some commonalities with numerical measurement. 
Analogies for relations (such as equality) and operations (such as addition)
do exist for both strings and graphs, but are richer in variety. 
For instance, appending one string to the end of another provides an analogy for addition, but one string could also be inserted anywhere within the other. Likewise two graphs can be joined by connecting any pair of nodes with an edge, or by identifying sets of nodes
to bring together with intermediate edges.
Moreover, strings and graphs also introduce other kinds of properties. 
Graphs, for instance, can exhibit clusters, modularity, and various topologies.
So observement opens up the prospect of formal systems with new kinds of operators.  
An enormous array of tools exists to support numerical measurement systems.  

Thirdly, observement has the power of mathematical abstraction.  We have illustrated this for graphs, which encapsulate the organization 
of large networks by relationships (edges) between entities, 
and in strings, which encapsulate extremely complex biological structures, such as proteins and DNA.

Lastly, observement can shape the development of theory and methods. 
Techniques for analysing and interpreting strings and graphs are now very active research areas. 
The example of motifs and other patterns, which are widely used  
to interpret strings and graphs, provide a case in point. 
A potential contribution of observement theory would be to encourage the development
of further methods and  applications based around widely used representations. 
In the early days of measurement theory,
only numeric data offered formal analytic methods of interpretation;
modern high performance computers and interactive visualization can effectively bring many kinds of observement systems within the scope of formal analysis. 

Finally, we point out that the examples we have given here 
are just the tip of the iceberg. 
There are many other kinds of data that already are, 
or could be observed using these representations.
There are also many other formal representations that could serve 
for certain kinds of non-numeric data.

\vskip1pc






\section*{Acknowledgment}
We would like to thank Professor John Crossley, Professor Graham Farr and Professor Mark Sanderson for useful suggestions on earlier versions of the manuscript. 

\bibliography{referencefile.bib}

\bibliographystyle{plain}

\end{document}